\documentclass{CSML}

\pdfoutput=1

\usepackage{lastpage}

\def\dOi{13(3:25)2017}
\lmcsheading%
{\dOi}
{1--\pageref{LastPage}}
{}
{}
{Apr.~15, 2015}
{Sep.~14, 2017}
{}

\usepackage{mathdots}
\usepackage{mdframed}
\usepackage{longtable}
\usepackage{graphicx}
\usepackage{amssymb, amsmath, amsxtra}
\usepackage{stmaryrd}
\usepackage{bbm}
\usepackage{galois}
\usepackage{listings}
\usepackage{color}
\usepackage[hidelinks]{hyperref}
\usepackage{times}

\usepackage{tikz}
\usetikzlibrary{arrows,backgrounds,shapes}
\newcommand{\comment}[1]{}

\def\ie{{\em i.e.}}

\newcommand{\Cc}{\mathcal{C}}

\newcommand{\cMs}{\mathcal{M}^*}

\def \tuple#1{\langle #1 \rangle}

\newcommand{\Ra}{\Rightarrow}

\newcommand{\ra}{\rightarrow}
\newcommand{\leqv}{\leq^v}

\DeclareMathOperator{\lb}{lb}
\DeclareMathOperator{\SL}{SL}

\DeclareMathOperator{\Ord}{Ord}
\newcommand{\triangleRel}{\mathrel{\overset\triangle\Longleftrightarrow}}
\newcommand{\ud}{\triangleq}
\newcommand{\da}{\ensuremath{\downarrow\!}}

\newtheorem{theorem}{Theorem}[section]
\newtheorem{lemma}[theorem]{Lemma}

\newtheorem{corollary}[theorem]{Corollary}
\newtheorem{example}[theorem]{Example}

\begin{document}

\title[A new characterization of complete Heyting and co-Heyting algebras]{A new characterization of complete Heyting and\\ co-Heyting algebras}

\address{Dipartimento di Matematica, University of Padova, Italy}
\email{francesco.ranzato@unipd.it}
\author[F.~Ranzato]{Francesco Ranzato}

\keywords{Complete Heyting algebra, Veinott ordering}
\subjclass{F.3.0 [Theory of computation]: Logics and meanings of programs---General}

\begin{abstract}
We give a new order-theoretic 
characterization of a complete Heyting and co-Hey\-ting algebra $C$. This result provides an unexpected 
relationship with the field of Nash equilibria, being
 based on the so-called Veinott
ordering relation on subcomplete sublattices of $C$, which is crucially 
used in Topkis' theorem for studying the order-theoretic stucture of 
Nash equilibria of supermodular games. 
\end{abstract}

\maketitle

\section*{Introduction}\label{intro}
Complete Heyting algebras~---~also called frames, while locales is used
for complete co-Heyting algebras~---~play a fundamental role as algebraic model of intuitionistic logic 
and in pointless topology \cite{johnstone1982,johnstone1983}. 
To the best of our knowledge, no characterization of complete Heyting and co-Heyting algebras has been known. 
As reported in \cite{bd}, a sufficient condition has been given in \cite{fun59} while a necessary condition
has been given by \cite{ch62}. 

We give here an order-theoretic 
characterization of complete Heyting and co-Hey\-ting algebras that puts forward an 
unexected relationship with Nash equilibria. Topkis' theorem \cite{topkis98} is well known in the theory of supermodular
games in mathematical
economics.  This result  shows that the set of solutions of a supermodular game, \ie, its set of pure-strategy Nash equilibria, 
is nonempty and contains a greatest element and a least one \cite{topkis78}. Topkis' theorem has been strengthned by 
\cite{zhou94}, where it is proved that this set of Nash equilibria is indeed a complete lattice. 
These results rely on so-called Veinott's ordering relation (also called strong set relation). 
Let $\tuple{C,\leq,\wedge,\vee}$ be a complete lattice. Then,
the relation $\leqv \subseteq \wp(C)\times \wp(C)$ on subsets of $C$, according
to Topkis~\cite{topkis78}, has been introduced by Veinott~\cite{topkis98,vei89}: for any $S,T\in \wp(C)$, 
$$S\leqv T \; \triangleRel\; \forall s\in S.\forall t\in T.\: s\wedge t\in S \;\:\&\;\: s\vee t \in T.$$
This relation $\leqv$ is always transitive and antisymmetric, while reflexivity
$S\leqv S$ holds if and only if $S$ is a sublattice of $C$. 
If $\SL(C)$ denotes the set of nonempty subcomplete sublattices of $C$ then $\tuple{\SL(C),\leqv}$ is therefore a poset. 
The proof of Topkis' theorem is then based on the fixed points of a certain 
mapping defined on the poset $\tuple{\SL(C),\leqv}$. 

To the best of our knowledge, no result is available on the order-theoretic properties of the Veinott 
poset $\tuple{\SL(C),\leqv}$. When is this poset a lattice? And a complete lattice? Our efforts in 
investigating these questions led to the following main result: the Veinott poset
$\SL(C)$ is a complete lattice if and only if $C$ is a complete Heyting and co-Heyting algebra. 
This finding therefore reveals an unexpected 
link between complete Heyting algebras and Nash equilibria of supermodular games. 
This characterization of the Veinott relation $\leqv$ could be exploited for generalizing a recent approach 
based on abstract interpretation for approximating the Nash equilibria of supermodular games introduced by~\cite{ran16}.

\section{Notation}
If $\tuple{P,\leq}$ is a poset and 
$S\subseteq P$ then $\lb(S)$ denotes the set of lower bounds of $S$, \ie, 
$\lb(S)\ud \{ x\in P~|~ \forall s\in S.\: x\leq s\}$, while 
if $x\in P$ then $\da x \ud \{y\in P~|~ y\leq x\}$. 

\noindent
Let $\tuple{C,\leq,\wedge,\vee}$ be a complete lattice. 
A nonempty subset $S\subseteq C$ is a subcomplete sublattice of $C$ if for all its nonempty subsets $X\subseteq S$, $\wedge X\in S$ 
and $\vee X \in S$, while $S$ is merely a sublattice of $C$ if this holds for all its  nonempty and finite subsets $X\subseteq S$ only. 
If $S \subseteq C$ then the nonempty Moore closure of $S$ is defined as
$\cMs (S) \ud \{\wedge X \in C~|~ X\subseteq S, X\neq \varnothing \}$. 
Let us observe that $\cMs$ is an upper closure operator on the poset $\tuple{\wp(C),\subseteq}$, meaning that: 
(1)~$S\subseteq T\: \Ra\: \cMs(S) \subseteq \cMs(T)$; (2)~$S\subseteq \cMs(S)$; (3)~$\cMs(\cMs(S)) = \cMs(S)$. 

\noindent
We define 
$$\SL(C) \triangleq \{S\subseteq C ~|~ S \neq \varnothing,\, S \text{~subcomplete sublattice of~} C\}.$$
Thus, if $\leqv$ denotes the Veinott ordering defined in Section~\ref{intro} then 
$\tuple{\SL(C),\leqv}$ is a poset.

\noindent
$C$ is a complete Heyting algebra (also called frame) 
if for any $x\in C$ and $Y\subseteq C$, $x\wedge (\bigvee Y) = \bigvee_{y\in Y} x\wedge y$,
while it is a complete co-Heyting algebra (also called locale) if the dual equation 
$x\vee (\bigwedge Y) = \bigwedge_{y\in Y} x\vee y$ holds. Let us recall that these two 
notions are orthogonal, for example the complete lattice of open subsets of $\mathbb{R}$ ordered by $\subseteq$ is
a complete Heyting algebra, but not a complete co-Heyting algebra. 
$C$ is (finitely) distributive if for any $x,y,z\in C$, $x\wedge(y\vee z) = (x\wedge y) \vee (x\wedge z)$. 
Let us also recall that 
$C$ is completely distributive if for any family $\{x_{j,k}~|~j\in J,k\in K(j)\}\subseteq C$,  we have that 
$$\bigwedge_{j\in J} \bigvee_{k\in K(j)} x_{j,k} = \bigvee_{f\in J\leadsto K} \bigwedge_{j\in J} x_{j,f(j)}$$
where $J$ and, 
for any $j\in J$, $K(j)$ are sets
of indices and
$J\leadsto K \ud \{f: J\ra \cup_{j\in J} K(j)~|~ \forall j\in J.\: f(j)\in K(j)\}$ denotes the set of choice functions. 
It turns out that the class of completely distributive complete lattices is strictly contained in the class
of complete Heyting and co-Heyting algebras. Clearly, any completely distribuitive lattice is a 
complete Heyting and co-Heyting algebra. On the other hand, this containment turns out to be strict, as
shown by the following counterexample.  

\begin{example}\rm
Let us recall that a subset $S\subseteq [0,1]$ of real numbers 
is a regular open set if $S$ is open and $S$ coincides with the interior
of the closure of $S$. For example, $(1/3,2/3)$ and $(0,1/3) \cup (2/3,1)$ are both regular open sets, while 
$(1/3,2/3) \cup (2/3,1)$ is open but not regular. 
 Let us 
consider $\mathcal{C}=\tuple{\{ S\subseteq [0,1]~|~ S ~\text{is a regular open set}\},\subseteq}$. It is known that 
$\mathcal{C}$ is a complete Boolean algebra (see e.g.\ \cite[Theorem~12, Section~2.5]{vla02}). 
As a consequence, $\mathcal{C}$ is a 
complete Heyting and co-Heyting algebra  (see e.g.\ \cite[Theorem~3, Section~0.2.3]{vla02}). 

\noindent
Recall that an element $a\in C$ in a complete lattice $C$ is an atom if $a$ is different from the least
element $\bot_C$ of $C$ and for any $x\in C$, if $\bot_C < x \leq a$ then $x=a$, while $C$ is atomic if 
for any $x\in C\smallsetminus\{\bot_C\}$ there exists an atom $a\in C$ such that $a\leq x$. 
It turns out that $\Cc$ does not have atoms: in fact, any regular open set $S\in \Cc$ is a union 
of open sets, namely, $S=\cup \{ U\subseteq [0,1]~|~
U ~\text{is open},\: U\subseteq S\}$ (see e.g.\ \cite[Section~2.5]{vla02}), 
so that no $S\in \Cc\smallsetminus \{\varnothing\}$ can be an atom of $\mathcal{C}$.
In turn, this implies that $\Cc$ is a complete Boolean algebra which is not atomic.  
It known that a complete Boolean algebra is completely distributive if and only if it is atomic
(see \cite[Theorem 14.5, Chapter 5]{hand}). Hence, since $\Cc$ is not atomic, we obtain that 
$\mathcal{C}$ is not completely distributive. 
\qed
\end{example}

\section{The Sufficient Condition}
To the best of our knowledge, no result is available on the order-theoretic properties of the Veinott 
poset $\tuple{\SL(C),\leqv}$. 
The following example shows that, in general, 
$\tuple{\SL(C),\leqv}$ is not a lattice. 

\begin{example}\label{exa1}\rm 
Consider the nondistributive pentagon lattice $N_5$, where, to use a compact notation,  
subsets of $N_5$ are denoted by strings of letters. 
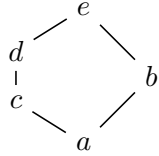
\begin{center}
    \begin{tikzpicture}[scale=0.45]
      \draw (0,4) node[name=e] {{$e$}};
      \draw (-2,2.75) node[name=d] {{$d$}};
      \draw (2,2) node[name=b] {{$b$}};
      \draw (-2,1.25) node[name=c] {{$c$}};
      \draw (0,0) node[name=a] {{$a$}};

      \draw[semithick] (a) -- (b);
      \draw[semithick] (a) -- (c);
      \draw[semithick] (c) -- (d);
      \draw[semithick] (b) -- (e);
      \draw[semithick] (d) -- (e);

\end{tikzpicture}
\end{center}
Consider $ed,abce \in \SL(N_5)$. It turns out that 
$\da ed =\{a,c,d,ab,ac,ad,cd,ed,acd,ade,cde,abde,$ $acde,abcde\}$ and 
$\da abce =\{a,ab,ac,abce\}$. Thus, $\{a,ab,ac\}$ is the set of common lower bounds of $ed$ and $abce$. 
However, the set $\{a,ab,ac\}$ does not include a greatest
element, since  $a\leqv ab$ and $a\leqv ac$ while $ab$ and $ac$ are incomparable. 
Hence, $ab$ and $c$ are maximal lower bounds of $ed$ and $abce$, so that $\tuple{\SL(N_5),\leqv}$ is not
a lattice. \qed
\end{example}

Indeed, the following result shows that if $\SL(C)$ turns out to be a lattice then $C$ must necessarily be distributive. 

\begin{lemma}\label{finite-lemma}
If $\tuple{\SL(C),\leqv}$ is a lattice then $C$ is distributive. 
\end{lemma}
\begin{proof}
By the basic characterization of distributive lattices, we know that $C$ is not distributive iff either the pentagon $N_5$ is a sublattice
of $C$ or
the diamond $M_3$ is a sublattice of $C$. We consider separately these two possibilities. 

\medskip
\noindent
$(N_5)$ Assume that $N_5$, as depicted by the diagram in Example~\ref{exa1}, is a sublattice of $C$. Following
Example~\ref{exa1}, we consider the sublattices $ed,abce\in  \tuple{\SL(C),\leqv}$ and we prove that their meet 
does not exist. 
By Example~\ref{exa1}, $ab,ac \in \lb(\{ed,abce\})$. Consider any $X\in \SL(C)$ such that $X\in \lb(\{ed,abce\})$. Assume that
$ab\leqv X$. If $x\in X$ then, by $ab\leqv X$, we have that 
$b\vee x \in X$. Moreover, by $X\leqv abce$, $b\vee x \in \{a,b,c,e\}$. If $b\vee x = e$ then we would have that 
$e\in X$,  and in turn, by $X\leqv ed$, $d=e\wedge d \in X$, so that, by $X\leqv abce$, we would get the contradiction 
$d=d\vee c \in \{a,b,c,e\}$. 
Also, if $b\vee x =c$ then we would have that 
$c\in X$,  and in turn, by $ab \leqv X$, $e=b\wedge c \in X$, so that, as in the previous case, we would get the contradiction 
$d=d\vee c \in \{a,b,c,e\}$. 
Thus, we necessarily have that $b\vee x\in \{a,b\}$. On the one hand, 
if $b\vee x =b$ then $x\leq b$ so that, by $ab \leqv X$, $x=b\wedge x\in \{a,b\}$. 
On the other hand, if $b\vee x = a$ then $x\leq a$ so that, by $ab \leqv X$, $x=a\wedge x\in \{a,b\}$. 
Hence, $X\subseteq \{a,b\}$. Since $X\neq \varnothing$, suppose that $a\in X$. Then, by $ab \leqv X$, 
$b=b\vee a\in X$. If, instead, $b\in X$ then, by $X\leqv abce$, $a=b\wedge a\in X$. We have therefore shown that $X=ab$.
An analogous argument shows that if $ac \leqv X$ then $X=ac$. If the meet of $ed$ and $abce$ would exist, call it $Z\in \SL(C)$, 
from $Z\in \lb(\{ed,abce\})$ and $ab,ac \leqv Z$ we would get the contradiction $ab=Z=ac$.  

\medskip
\noindent
$(M_3)$ Assume that the diamond $M_3$, as depicted by the following diagram, is a sublattice of $C$. 
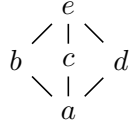
\begin{center}
    \begin{tikzpicture}[scale=0.35]
      \draw (0,0) node[name=a] {{$a$}};
     \draw (-2,2) node[name=b] {{$b$}};
       \draw (0,2) node[name=c] {{$c$}};
       \draw (2,2) node[name=d] {{$d$}};  
      \draw (0,4) node[name=e] {{$e$}};

      \draw[semithick] (a) -- (b);
      \draw[semithick] (a) -- (c);
      \draw[semithick] (a) -- (d);
      \draw[semithick] (e) -- (b);
      \draw[semithick] (e) -- (c);
      \draw[semithick] (e) -- (d);

\end{tikzpicture}
\end{center}
In this case, we consider the sublattices $eb,ec\in  \tuple{\SL(C),\leqv}$ and we prove that their meet 
does not exist. 
It turns out that
 $abce,abcde \in \lb(\{eb,ec\})$ while $abce$ and $abcde$ are incomparable. 
 Consider any $X\in \SL(C)$ such that $X\in \lb(\{eb,ec\})$. Assume that
$abcde\leqv X$. If $x\in X$ then, by $X \leqv eb,ec$, we have that 
$x\wedge b, x\wedge c \in X$, so that $x\wedge b \wedge c = x\wedge a\in X$. 
{}From $abcde \leqv X$, we obtain that for any $y\in \{a,b,c,d,e\}$, $y=y\vee (x\wedge a) \in X$. Hence, $\{a,b,c,d,e\}
\subseteq X$. From $X\leqv eb$, we derive that $x\vee b\in \{e,b\}$, and, from $abcde \leqv X$, 
we also have that $x\vee b\in X$. If $x\vee b=e$ then $x\leq e$, so that, from $abcde\leqv X$, we obtain
$x=e\wedge x\in \{a,b,c,d,e\}$. If, instead, $x\vee b = b$ then $x\leq b$,  
so that, from $abcde\leqv X$, we derive
$x=b\wedge x\in \{a,b,c,d,e\}$. In both cases, we have that $X\subseteq \{a,b,c,d,e\}$. We thus conclude that $X=abcde$. 
An analogous argument shows that if $abce \leqv X$ then $X=abce$. Hence, similarly to the previous case $(N_5)$, 
the meet of $eb$ and $ec$ does not exist. 
\end{proof}

Moreover, we show that if we require $\SL(C)$ to be a complete lattice then the complete lattice $C$ must be 
a complete Heyting and co-Heyting algebra. 
Let us remark that this proof makes use of the axiom of choice. 
\begin{theorem}\label{main-th}
If $\tuple{\SL(C),\leqv}$ is a complete lattice then $C$ is a
complete Heyting and co-Heyting algebra. 
\end{theorem}
\begin{proof}
Assume that the complete lattice $C$ is not a complete co-Heyting algebra. 
If $C$ is not distributive, then, by Lemma~\ref{finite-lemma},
${\tuple{\SL(C),\leqv}}$ is not a complete lattice. 
Thus, let us assume that $C$ is distributive. 
The (dual) characterization in \cite[Remark 4.3, p.~40]{gierz} states that a complete lattice $C$ is 
a complete co-Heyting algebra
iff $C$ is distributive and join-continuous (\ie, the join distributes over arbitrary meets of directed subsets). 
Consequently, it turns out that $C$ is not join-continuous. 
Thus, by the result in \cite{bruns67} on directed sets and chains (see also \cite[Exercise~4.9, p.~42]{gierz}), 
there exists an infinite descending chain $\{a_\beta\}_{\beta<\alpha}\subseteq C$, for some ordinal $\alpha\in \Ord$, such that 
if $\beta<\gamma<\alpha$ then $a_\beta > a_\gamma$, and an element $b\in C$ such that 
$\bigwedge_{\beta< \alpha} a_\beta \leq b  < \bigwedge_{\beta< \alpha} (b\vee a_\beta)$. 
We observe the following facts:
\begin{enumerate}[label=(\bf\Alph*)]
\item[(A)]  
$\alpha$ must necessarily be a limit ordinal (so that 
$|\alpha|\geq |\mathbb{N}|$), otherwise if $\alpha$ is a successor ordinal then 
we would have that, for any $\beta < \alpha$, 
$a_{\alpha -1}\leq a_\beta$, so that $\bigwedge_{\beta< \alpha} a_\beta = a_{\alpha-1} \leq  b$, and in turn
 we would obtain 
$\bigwedge_{\beta< \alpha} (b\vee a_\beta)= b\vee a_{\alpha-1}=b$, \ie,  a contradiction.
\item[(B)]
 We have that 
$\bigwedge_{\beta< \alpha} a_\beta < b$, 
otherwise $\bigwedge_{\beta< \alpha} a_\beta = b$ would imply that $b\leq a_\beta$ for any $\beta < \alpha$, so that 
$\bigwedge_{\beta< \alpha} (b\vee a_\beta)=\bigwedge_{\beta< \alpha} a_\beta =b$, which is a contradiction. 
\item[(C)]

Firstly, observe that $\{b \vee a_\beta\}_{\beta< \alpha}$ is an infinite descending chain in $C$. 
Let us consider a limit ordinal $\gamma < \alpha$. 
Without loss of generality, we assume that the glb's of the subchains
$\{a_\rho\}_{\rho < \gamma}$ and $\{b \vee a_\rho\}_{\rho < \gamma}$ belong, respectively, 
to the chains $\{a_\beta\}_{\beta<\alpha}$ and  $\{b \vee a_\beta\}_{\beta< \alpha}$. 
For our purposes, this is not a restriction because
the elements 
$\bigwedge_{\rho < \gamma} a_\rho$ and $\bigwedge_{\rho<\gamma} (b \vee a_\rho)$ can be added to the respective chains
$\{a_\beta\}_{\beta<\alpha}$ and $\{b \vee a_\beta\}_{\beta< \alpha}$ and these extensions 
would preserve both the glb's of the chains $\{a_\beta\}_{\beta<\alpha}$ and $\{b \vee a_\beta\}_{\beta< \alpha}$ and
the inequalities $\bigwedge_{\beta< \alpha} a_\beta < b  < \bigwedge_{\beta< \alpha} (b\vee a_\beta)$.
Hence, by this nonrestrictive assumption, we have that for any limit ordinal $\gamma < \alpha$, 
$\bigwedge_{\rho < \gamma} a_\rho = a_\gamma$ and 
$\bigwedge_{\rho<\gamma} (b \vee a_\rho) = b\vee a_\gamma$ hold. 

\item[(D)] 
Let us consider the set  $S=\{a_\beta~|~ \beta <\alpha,\: 
\forall \gamma \geq \beta.\: b \not\leq a_\gamma\}$. 
Then, $S$ must be nonempty, otherwise we would have that for any $\beta<\alpha$ there exists some $\gamma_\beta\geq \beta$ such that 
$b\leq a_{\gamma_\beta} \leq a_\beta$, and this would imply that for any $\beta <\alpha$, 
$b\vee a_\beta=a_\beta$, so that we would obtain 
$\bigwedge_{\beta<\alpha} (b \vee a_\beta) =
\bigwedge_{\beta<\alpha} a_\beta$, which is a contradiction. Since any chain  in (\ie, subset of) 
$S$ has an upper bound in $S$, by Zorn's Lemma, $S$ contains the maximal element $a_{\bar{\beta}}$, for some $\bar{\beta}<\alpha$, such
that for any $\gamma< \alpha$ and $\gamma \geq \bar{\beta}$, $b \not\leq a_\gamma$. We also observe that $\bigwedge_{\beta < \alpha} a_\beta
=\bigwedge_{\bar{\beta}\leq \gamma < \alpha} a_\gamma$ and  $\bigwedge_{\beta < \alpha} (b \vee a_\beta)
=\bigwedge_{\bar{\beta}\leq \gamma < \alpha} (b \vee a_\gamma)$. Hence, without loss of generality, we assume that
the chain $\{a_\beta\}_{\beta<\alpha}$ is such that, for any $\beta<\alpha$, 
$b \not\leq a_\beta$ holds. 
\end{enumerate}
 
\noindent For any ordinal $\beta < \alpha$~---~therefore, we remark that 
the limit ordinal $\alpha$ is not included~---~we define, by transfinite induction, the following subsets 
$X_\beta \subseteq C$:
\begin{itemize}
\item[--] $\beta=0$ $\;\Ra\;$ $X_\beta \ud \{a_0,\: b\vee a_0\}$;
\item[--] $\beta>0$ $\;\Ra\;$ $ X_\beta \ud \bigcup_{\gamma<\beta}X_\gamma \cup \{b \vee a_{\beta}\}\cup 
\{(b \vee a_{\beta}) \wedge a_\delta~|~ \delta\leq \beta\}$.
\end{itemize}
Observe that, for any $\beta >0$,  $(b \vee a_{\beta}) \wedge a_\beta = a_\beta$ and that the set
$\{b \vee a_{\beta}\}\cup 
\{(b \vee a_{\beta}) \wedge a_\delta~|~ \delta\leq \beta\}$ is indeed a chain. Moreover, 
if $\delta\leq \beta$ then, 
 by distributivity, we have that $(b \vee a_{\beta}) \wedge a_\delta = 
(b \wedge a_{\delta}) \vee (a_{\beta} \wedge a_\delta) =(b\wedge a_\delta)\vee a_{\beta}$. Moreover,
if $\gamma < \beta < \alpha$ then $X_\gamma \subseteq X_\beta$. 
 
We show, by transfinite induction on $\beta$, that  for any $\beta< \alpha$, $X_\beta \in \SL(C)$. 
Let $\delta\leq \beta$ and $\mu\leq \gamma <\beta$. We notice the following facts:
\begin{enumerate}[widest*=10]
\item $(b\vee a_\beta)\wedge (b\vee a_\gamma) = b\vee a_\beta \in X_\beta$
\item $(b\vee a_\beta)\vee (b\vee a_\gamma) = b\vee a_\gamma \in X_\gamma \subseteq X_\beta$
\item $(b \vee a_{\beta}) \wedge \big((b \vee a_{\gamma}) \wedge a_\mu\big) =  
(b \vee a_{\beta}) \wedge a_\mu \in X_{\beta}$
\item 
$(b \vee a_{\beta}) \vee \big((b \vee a_{\gamma}) \wedge a_\mu\big) =
(b \vee a_{\beta}) \vee (b \wedge a_{\mu}) \vee a_\gamma =  b \vee a_{\gamma} \in  X_{\gamma}\subseteq X_\beta$

\item $\big((b \vee a_{\beta})\wedge a_\delta\big) \wedge \big((b \vee a_{\gamma}) \wedge a_\mu\big) =  
(b \vee a_{\beta}) \wedge a_{\max(\delta,\mu)} \in X_{\beta}$

\item $\big((b \vee a_{\beta})\wedge a_\delta\big)  
\vee \big((b \vee a_{\gamma}) \wedge a_\mu\big) =
\big((b \wedge a_{\delta})\vee a_\beta\big)  \vee \big((b \wedge a_{\mu}) \vee a_\gamma\big) = 
 (b \wedge a_{\min(\delta,\mu)}) \vee a_{\gamma} = (b \vee a_{\gamma} ) 
 \wedge a_{\min(\delta,\mu)}
 \in  X_{\gamma}\subseteq X_\beta$

\item if $\beta$ is a limit ordinal then, by point~(C) above, $ \bigwedge_{\rho < \beta} (b\vee a_\rho) = b\vee a_\beta$
holds; therefore, 
$\bigwedge_{\rho < \beta} \big((b\vee a_\rho) \wedge a_\delta \big)= \big( \bigwedge_{\rho < \beta} (b\vee a_\rho)\big) \wedge 
a_\delta = (b\vee a_\beta)\wedge a_\delta\in X_\beta$; in turn, by taking the glb of these latter elements in $X_\beta$,
we have that
$\bigwedge_{\delta \leq \beta} \big((b\vee a_\beta) \wedge a_\delta\big)
= (b\vee a_\beta) \wedge \big(\bigwedge_{\delta \leq \beta} a_\delta\big) =(b\vee a_\beta) \wedge a_\beta= 
a_\beta \in X_\beta$ 

\end{enumerate}
Since $X_0 \in \SL(C)$ obviously holds, the points (1)-(7) above show, by transfinite induction, 
that for any $\beta <\alpha$, $X_\beta$ is closed under arbritrary lub's and glb's of nonempty subsets, 
\ie, $X_\beta\in \SL(C)$. 
In the following, we prove that the glb of $\{X_\beta\}_{\beta < \alpha}\subseteq \SL(C)$ in $\tuple{\SL(C),\leqv}$ does not exist. 

Recalling, by point~(A) above, that $\alpha$ is a limit ordinal, we define
$A\ud \cMs(\bigcup_{\beta<\alpha} X_\beta)$. By point~(C) above, we observe that for any limit ordinal 
$\gamma<\alpha$, the $\bigcup_{\beta<\alpha} X_\beta$ already contains the glb's 
$$\bigwedge_{\rho<\gamma} (b\vee a_\rho) = b \vee a_\gamma \in X_\gamma,\qquad 
\bigwedge_{\rho<\gamma} a_\rho = a_\gamma \in X_\gamma,$$
$$\{ \big(\bigwedge_{\rho<\gamma} (b\vee a_\rho)\big) \wedge a_\delta ~|~ \delta <\gamma\}
=\{ (b\vee a_\gamma) \wedge a_\delta ~|~ \delta <\gamma\}\subseteq X_\gamma.$$
Hence, by taking the glb's of all the chains in $\bigcup_{\beta<\alpha} X_\beta$, $A$ turns out to be
as follows:
$$A = \bigcup_{\beta<\alpha} X_\beta \cup \{\bigwedge_{\beta<\alpha} (b\vee a_\beta),\: 
\bigwedge_{\beta<\alpha} a_\beta \}\cup
\{ \big(\bigwedge_{\beta<\alpha} (b\vee a_\beta)\big) \wedge a_\delta ~|~ \delta <\alpha\}.$$
Let us show that $A\in \SL(C)$. 
First, we observe that $\bigcup_{\beta<\alpha} X_\beta$ is closed under arbitrary nonempty lub's. In fact, if $S\subseteq  
\bigcup_{\beta<\alpha} X_\beta$ then $S = \bigcup_{\beta<\alpha} (S\cap X_\beta)$, so that
$$ \bigvee S = \bigvee \bigcup_{\beta<\alpha} (S\cap X_\beta)
= \bigvee_{\beta<\alpha} \bigvee S\cap X_\beta.$$
Also, if $\gamma <\beta <\alpha$ 
then $S\cap X_\gamma \subseteq S\cap X_\beta$
and, in turn, $\bigvee S\cap X_\gamma \leq 
\bigvee S\cap X_\beta$,  so that $\{\bigvee S\cap X_\beta\}_{\beta< \alpha}$ is an increasing chain.
Hence, since  $\bigcup_{\beta<\alpha}X_\beta$ 
does not contain infinite increasing chains, there exists some $\gamma < \alpha$ such that $\bigvee_{\beta<\alpha} \bigvee S\cap X_\beta
= \bigvee S\cap X_\gamma \in X_\gamma$, and consequently $\bigvee S \in \bigcup_{\beta<\alpha}X_\beta$. 
Moreover, $\{\big(\bigwedge_{\beta<\alpha} (b\vee a_\beta)\big) \wedge a_\delta\}_{\delta <\alpha}\subseteq 
A$ is a chain 
whose lub is $\big(\bigwedge_{\beta<\alpha} (b\vee a_\beta)\big) \wedge a_0$ which belongs to the chain itself, 
while
its glb is 
$$\bigwedge_{\delta< \alpha} \big(\bigwedge_{\beta<\alpha} (b\vee a_\beta)\big) \wedge a_\delta = 
\big(\bigwedge_{\beta<\alpha} (b\vee a_\beta)\big) \wedge \bigwedge_{\delta< \alpha} a_\delta = \bigwedge_{\delta< \alpha} a_\delta\in A.$$
Finally, if $\delta \leq \gamma < \alpha$ then we have that:
\begin{enumerate}[resume,widest*=10]
\item $\big(\bigwedge_{\beta<\alpha} (b\vee a_\beta)\big)\wedge (b\vee a_\gamma) = \bigwedge_{\beta<\alpha} (b\vee a_\beta) \in A$
\item $\big(\bigwedge_{\beta<\alpha} (b\vee a_\beta)\big)\vee (b\vee a_\gamma) = b\vee a_\gamma \in X_\gamma \subseteq A$
\item $\big(\bigwedge_{\beta<\alpha} (b\vee a_\beta)\big) \wedge \big((b \vee a_{\gamma}) \wedge a_\delta\big) =  
\big(\bigwedge_{\beta<\alpha} (b\vee a_\beta)\big) \wedge a_\delta \in A$
\item We have that
$\big(\bigwedge_{\beta<\alpha} (b\vee a_\beta)\big) \vee \big((b \vee a_{\gamma}) \wedge a_\delta\big) =
\big(\bigwedge_{\beta<\alpha} (b\vee a_\beta)\big) \vee (b \wedge a_{\delta}) \vee a_\gamma
=  \big(\bigwedge_{\beta<\alpha} (b\vee a_\beta)\big) \vee a_\gamma$. Moreover, 
$
b\vee a_\gamma \leq \big(\bigwedge_{\beta<\alpha} (b\vee a_\beta)\big) \vee a_\gamma \leq (b \vee a_\gamma) \vee a_\gamma
= b\vee a_\gamma
$; hence, $\big(\bigwedge_{\beta<\alpha} (b\vee a_\beta)\big) \vee \big((b \vee a_{\gamma}) \wedge a_\delta\big)
= b\vee a_\gamma \in  X_{\gamma}\subseteq A$.
\end{enumerate}
Summing up, we have therefore shown that $A\in \SL(C)$. 

We now prove that $A$ is a lower bound of $\{X_\beta\}_{\beta < \alpha}$, \ie, we prove, by transfinite induction 
on $\beta$, that for any $\beta <\alpha$,
$A\leqv X_\beta$. 
\begin{itemize}
\item $\big(A\leqv X_0\big)$:  this is a consequence of the following easy equalities, for any $\delta\leq \beta<\alpha$: 
$(b \vee a_\beta)\wedge a_0 \in X_\beta \subseteq A$; $(b \vee a_\beta)\vee a_0 = b \vee a_0 \in X_0$; 
$(b \vee a_\beta)\wedge (b \vee a_0) =b\vee a_\beta \in X_\beta \subseteq A$; $(b \vee a_\beta)\vee (b \vee a_0) = b \vee a_0 \in X_0$; 
$\big((b \vee a_\beta)\wedge a_\delta\big) \wedge a_0 = (b \vee a_\beta)\wedge a_\delta \in X_\beta \subseteq A$; 
$\big((b \vee a_\beta)\wedge a_\delta\big) \vee a_0 =  a_0  \in X_0$; 
$\big((b \vee a_\beta)\wedge a_\delta\big) \wedge (b \vee a_0) = (b \vee a_\beta)\wedge a_\delta \in X_\beta \subseteq A$; 
$\big((b \vee a_\beta)\wedge a_\delta\big) \vee (b \vee a_0) =  b \vee a_0  \in X_0$.

\item $\big(A\leqv X_\beta$, $\beta>0\big)$: Let $a\in A$ and $x\in X_\beta$. If $x\in \bigcup_{\gamma<\beta} X_\gamma$
then
$x\in X_\gamma$ for some $\gamma<\beta$, so that, since by inductive hypothesis $A\leqv X_\gamma$, 
we have that $a\wedge x\in A$ and $a\vee x\in X_\gamma\subseteq X_\beta$. Thus, assume that $x\in X_\beta \smallsetminus \big(\bigcup_{\gamma<\beta} X_\gamma\big)$. If $a\in X_\beta$ then $a\wedge x \in X_\beta\subseteq A$ and $a \vee x\in X_\beta$. If $a\in X_\mu$, for some $\mu > \beta$, 
then $a\wedge x \in X_\mu \subseteq A$, while points~(2), (4) and (6) above show that $a \vee x \in X_\beta$. 
If $a=\bigwedge_{\beta<\alpha} (b\vee a_\beta)$ then points~(8)-(11) above show that $a\wedge x \in A$ and $a\vee x \in X_\beta$. 
If $a= \big(\bigwedge_{\gamma <\alpha} (b\vee a_\gamma)\big) \wedge a_\mu$, for some $\mu <\alpha$, and $\delta\leq \beta$ then we have that: 
\begin{enumerate}[resume,widest*=10]
\item $\big(\big(\bigwedge_{\gamma<\alpha} (b\vee a_\gamma)\big) \wedge a_\mu\big) \wedge (b\vee a_\beta) = 
\big(\bigwedge_{\gamma<\alpha} (b\vee a_\gamma)\big) \wedge a_\mu   \in A$

\item $\big(\big(\bigwedge_{\gamma<\alpha} (b\vee a_\gamma)\big) \wedge a_\mu\big) \vee (b\vee a_\beta) =
\big(\big(\bigwedge_{\gamma<\alpha} (b\vee a_\gamma)\big) \vee (b\vee a_\beta)\big) \wedge (a_\mu \vee (b \vee a_\beta))=
(b\vee a_\beta) \wedge (b \vee a_{\min(\mu,\beta)}) = b \vee a_\beta
\in X_\beta$

\item $\big(\big(\bigwedge_{\gamma<\alpha} (b\vee a_\gamma)\big) \wedge a_\mu\big)  \wedge \big((b \vee a_{\beta}) \wedge a_\delta\big) =  
\big(\bigwedge_{\gamma<\alpha} (b\vee a_\gamma)\big) \wedge a_{\max(\mu,\delta)} \in A$

\item 
\begin{align*}
\big(\big(\bigwedge_{\gamma<\alpha} (b\vee a_\gamma)\big) \wedge a_\mu\big) \vee \big((b \vee a_{\beta}) \wedge a_\delta\big) & =\\
\big(\big(\bigwedge_{\gamma<\alpha} (b\vee a_\gamma)\big) \vee (b \vee a_{\beta})\big) \wedge 
\big(\big(\bigwedge_{\gamma<\alpha} (b\vee a_\gamma)\big) \vee a_\delta\big) \wedge \big(a_\mu \vee (b\vee a_\beta)\big) \wedge (a_\mu \vee a_\delta)&=\\
(b \vee a_\beta) \wedge 
(b \vee a_\delta) \wedge 
\big(b\vee a_{\min(\mu,\beta)}\big) \wedge a_{\min(\mu,\delta)}&=\\
(b \vee a_\beta) \wedge a_{\min(\mu,\delta)}&\in X_\beta \end{align*}
\end{enumerate}
Finally, if $a = \bigwedge_{\gamma < \alpha} a_\gamma$ and $x\in X_\beta$ then $a\leq x$ so that $a \wedge x = a \in A$ and 
$a \vee x = x \in X_\beta$. Summing up, we have shown that $A\leqv X_\beta$. 
\end{itemize}

\noindent Let us now prove that $b\not\in A$. Let us first observe that for any $\beta < \alpha$, we have that $a_\beta\not\leq b$: in fact, 
if $a_\gamma \leq b$, for some $\gamma <\alpha$ then, for any $\delta\leq \gamma$, $b\vee a_\delta = b$, so that
we would obtain $\bigwedge_{\beta <\alpha} (b\vee a_\beta) = b$, which is a contradiction. Hence, for any $\beta < \alpha$ and 
$\delta \leq \beta$, it turns out that
$b\neq b\vee a_\beta$ and $b \neq (b \wedge a_\delta) \vee a_\beta = (b \vee a_\beta)\wedge a_\delta$. Moreover, by point~(B) above,
$b \neq \bigwedge_{\beta <\alpha} (b\vee a_\beta)$, while, by hypothesis, 
$b \neq \bigwedge_{\beta< \alpha} a_\beta$. Finally, for any $\delta <\alpha$, if
$b = \big(\bigwedge_{\beta <\alpha} (b\vee a_\beta)\big)\wedge a_\delta$ then we would
derive that $b\leq a_\delta$, which, by point~(D) above, is a contradiction. 

Now, we define $B \ud \cMs(A\cup \{b\})$, so that
$$B= A \cup \{b\}\cup \{b\wedge a_\delta~|~\delta< \alpha\}.$$
Observe that for any $a\in A$, with $a\neq \bigwedge_{\beta<\alpha} a_\beta$, and for any $\delta<\alpha$, we have that
$b\wedge a_\delta \leq a$, while 
$b \vee \Big(\big(\bigwedge_{\beta <\alpha} (b\vee a_\beta)\big)\wedge a_\delta\Big)
= \Big(b \vee \big(\bigwedge_{\beta <\alpha} (b\vee a_\beta)\big)\Big) \wedge (b \vee a_\delta) = 
\big(\bigwedge_{\beta <\alpha} (b\vee a_\beta)\big) \wedge (b \vee a_\delta) = \bigwedge_{\beta <\alpha} (b\vee a_\beta)
\in B$. Also, for any $\delta \leq \beta < \alpha$, we have that 
$b \vee \big((b \vee a_\beta) \wedge a_\delta\big) =
\big(b \vee (b\vee a_\beta)\big) \wedge (b \vee a_\delta) = b \vee a_\delta \in B$. 
Also, 
$b \vee \big(\bigwedge_{\beta <\alpha} (b\vee a_\beta)\big) = \bigwedge_{\beta <\alpha} (b\vee a_\beta)\in B$
and  $b \vee \bigwedge_{\beta<\alpha}a_\beta 
=b\in B$. We have thus checked that $B$ is closed under lub's (of arbitrary nonempty subsets), 
\ie, $B\in \SL(C)$. Let us check that $B$ is a lower bound of $\{X_\beta\}_{\beta <\alpha}$. Since we have already shown that 
$A$ is a lower bound, and since $b\wedge a_\delta \leq b$, for any $\delta <\alpha$, it is enough to observe that for 
any $\beta<\alpha$ and $x\in X_\beta$, $b\wedge x \in B$ and $b\vee x \in X_\beta$. 
The only nontrivial case is for $x= (b\vee a_\beta)\wedge a_\delta$, for some $\delta\leq \beta<\alpha$. 
On the one hand, $b \wedge \big((b\vee a_\beta)\wedge a_\delta \big)= b\wedge a_\delta\in B$, on the other hand,
$b \vee \big((b\vee a_\beta)\wedge a_\delta \big)=b \vee \big((b\wedge a_\delta)\vee a_\beta \big)=
b \vee a_\beta \in X_\beta$.

Let us now assume that there exists $Y\in \SL(C)$ such that $Y$ is the glb of $\{X_\beta\}_{\beta <\alpha}$ in $\tuple{\SL(C),\leqv}$. Therefore, 
since we proved that 
$A$ is a lower bound, we have that $A\leqv Y$. Let us consider $y\in Y$. Since $b\vee a_0\in A$, we have that $b\vee a_0 \vee y\in Y$. 
Since $Y\leqv X_0 = \{a_0, b\vee a_0\}$, we have that $b \vee a_0 \vee y \vee a_0 = b\vee a_0 \vee y \in \{a_0,b\vee a_0\}$. 
If $b\vee a_0 \vee y = a_0$ then $b\leq a_0$, which, by point~(D), is a contradiction. Thus, we have that 
$b\vee a_0\vee y = b \vee a_0$, so that $y\leq b \vee a_0$
and $b \vee a_0 \in Y$. We know that if $x\in X_\beta$, for some $\beta < \alpha$, then $x \leq b\vee a_0$, so that, from
$Y\leqv X_\beta$, we obtain that $(b\vee a_0)\wedge x = x \in Y$, that is, $X_\beta \subseteq Y$. Thus, we have that
$\bigcup_{\beta <\alpha} X_\beta \subseteq Y$, 
and, in turn, by subset monotonicity of $\cMs$, we get $A=\cMs(\bigcup_{\beta <\alpha} X_\beta) \subseteq \cMs(Y)=Y$. Moreover, 
from $y\leq b \vee a_0$, since $A\leqv Y$ and $b\vee a_0 \in A$, we obtain $(b\vee a_0)\wedge y = y \in A$, that is $Y\subseteq A$. 
We have therefore shown that $Y=A$. Since we proved that 
$B$ is a lower bound, $B\leqv Y=A$
must hold. However, it turns out that $B \leqv A$ is a contradiction: by considering
$b\in B$ and $\bigwedge_{\beta< \alpha} a_\beta\in A$, we would have that
$b \vee \big(\bigwedge_{\beta< \alpha} a_\beta\big) = b \in A$, while we have shown above that $b\not\in A$. 
We have therefore shown that the glb of $\{X_\beta\}_{\beta <\alpha}$ in $\tuple{\SL(C),\leqv}$ does not exist. 

To close the proof, it is enough to observe that if $\tuple{C,\leq}$ is not a complete Heyting
algebra then, by duality, $\tuple{\SL(C),\leqv}$ does not have
lub's. 
\end{proof}

\section{The Necessary Condition}

It turns out that the property of being a complete lattice for the poset $\tuple{\SL(C),\leqv}$  is a 
necessary condition for a complete Heyting and co-Heyting algebra $C$. 

\begin{theorem}
If $C$ is a complete Heyting and co-Heyting algebra then
$\tuple{\SL(C),\leqv}$ is a complete lattice.  
\end{theorem}
\begin{proof}
Let $\{A_i\}_{i\in I} \subseteq \SL(C)$, for some family of indices $I\neq \varnothing$. Let us define
$$G \ud \{x\in \cMs(\cup_{i\in I} A_i) ~|~ \forall k\in I.\; \cMs(\cup_{i\in I} A_i) \:\cap \da x \leqv A_k\}.$$
The following three points show that $G$ is the glb of $\{A_i\}_{i\in I}$ in $\tuple{\SL(C),\leqv}$. 

\medskip
\noindent
$(1)$ We show that $G\in \SL(C)$. Let $\bot \ud \bigwedge_{i\in I} \bigwedge A_i$. First, 
$G$ is nonempty because it turns out that $\bot \in G$. 
Since, for any $i\in I$, $\bigwedge A_i \in A_i$ and $I\neq \varnothing$, 
we have that $\bot \in \cMs(\cup_i A_i)$. Let $y\in \cMs(\cup_i A_i) \:\cap \da 
\bot$ and,
for some $k\in I$,  
$a\in A_k$. On the one hand, we have that $y\wedge a \in \cMs(\cup_i A_i) \:\cap \da \bot$ trivially holds. 
On the other hand, since $y\leq \bot \leq a$, we have that $y\vee a = a \in A_k$. 

Let us now consider a set $\{x_j\}_{j\in J} \subseteq G$, for some family of indices $J\neq \varnothing$, so that, for any $j\in J$
and $k\in I$, $\cMs(\cup_i A_i) \:\cap \da x_j \leqv A_k$. 

First, notice that $\bigwedge_{j\in J} x_j 
\in \cMs(\cup_i A_i)$ holds. Then, since $\da  (\bigwedge_{j\in J} x_j) = \bigcap_{j\in J} \da x_j$ holds, 
we have that $\cMs(\cup_i A_i) \:\cap \da  (\bigwedge_{j\in J} x_j) = \cMs(\cup_i A_i) \,\cap (\bigcap_{j\in J}\! \da x_j)$, 
so that, for any $k \in I$, $\cMs(\cup_i A_i) \:\cap \da  (\bigwedge_{j\in J} x_j) \leqv A_k$, that is, 
$\bigwedge_{j\in J} x_j \in G$. 

Let us now prove that $\bigvee_{j\in J} x_j 
\in \cMs(\cup_i A_i)$ holds. First, since any $x_j\in \cMs(\cup_{i\in I} A_i)$, we have that $x_j = \bigwedge_{i\in K(j)} a_{j,i}$,
where, for any $j\in J$, $K(j)\subseteq I$ is a nonempty family of indices in $I$ 
such that for any $i\in K(j)$, $a_{j,i}\in A_i$. 
For any $i\in I$, we then define the family of indices $L(i)\subseteq J$ as follows: 
$L(i) \ud \{ j\in J~|~i\in K(j)\}$. Observe that it may happen that $L(i) = \varnothing$. 
Since for any $i\in I$ such that $L(i)\neq \varnothing$, 
$\{ a_{j,i}\}_{j\in L(i)}\subseteq A_i$ and $A_i$ is meet-closed,
we have that if $L(i)\neq \varnothing$ then $\hat{a}_i \ud \bigwedge_{l\in L(i)} a_{l,i}\in A_i$. 
Since, given $k\in I$ such that $L(k)\neq \varnothing$, for any $j\in J$, $\cMs(\cup_{i\in I} A_i) \:\cap \da x_j \leqv A_k$, we have that
for any $j\in J$, $x_j \vee \hat{a}_k \in A_k$. Since $A_k$ is join-closed, we obtain that $\bigvee_{j\in J} (x_j \vee \hat{a}_k) 
= (\bigvee_{j\in J} x_j) \vee \hat{a}_k \in A_k$. Consequently, 
$$
\bigwedge_{k\in I,\atop{L(k) \neq \varnothing}} \big( (\bigvee_{j\in J} x_j) \vee \hat{a}_k \big)\in \cMs(\cup_{i\in I} A_i).
$$
Since $C$ is a complete co-Heyting algebra,  
$$
\bigwedge_{k\in I,\atop{L(k) \neq \varnothing}} \big( (\bigvee_{j\in J} x_j) \vee \hat{a}_k \big)
= (\bigvee_{j\in J} x_j) \vee (\bigwedge_{k\in I,\atop{L(k) \neq \varnothing}} \hat{a}_k).
$$
Thus, since, for any $j\in J$, 
$$
\displaystyle \bigwedge_{k\in I,\atop{L(k) \neq \varnothing}} \hat{a}_k = \bigwedge_{j\in J} {\textstyle
\bigwedge_{i\in K(j)}} a_{j,i}\leq  x_j,
$$
we obtain that $ (\bigvee_{j\in J} x_j) \vee \displaystyle(\bigwedge_{k\in I,\atop{L(k) \neq \varnothing}} \hat{a}_k) =$ $\bigvee_{j\in J} x_j$,
so that $\bigvee_{j\in J} x_j \in \cMs(\cup_{i\in I} A_i)$. 

Finally, in order to prove that $\bigvee_{j\in J} x_j \in G$, let us
show that for any $k \in I$, $\cMs(\cup_i A_i) \:\cap \da  (\bigvee_{j\in J} x_j) \leqv A_k$. 
Let $y\in \cMs(\cup_i A_i) \:\cap \da  (\bigvee_{j\in J} x_j)$ and $a\in A_k$. For any $j\in J$, 
$y\wedge x_j \in \cMs(\cup_i A_i) \:\cap \da  (\bigvee_{j\in J} x_j)$, so that $(y\wedge x_j)\vee a\in A_k$. 
Since $A_k$ is join-closed, we obtain that $\bigvee_{j\in J} \big((y\wedge x_j)\vee a\big) =a \vee \big(\bigvee_{j\in J} (y\wedge x_j)\big) \in A_k$. Since $C$ is a complete Heyting algebra, 
$a \vee \big(\bigvee_{j\in J} (y\wedge x_j)\big) = a \vee \big(y \wedge  
(\bigvee_{j\in J} x_j) \big)$. Since $y \wedge  
(\bigvee_{j\in J} x_j) = y$, we derive that $y \vee a \in A_k$. On the other hand, $y\wedge a\in \cMs(\cup_i A_i) \:\cap \da  (\bigvee_{j\in J} x_j)$
trivially holds. 

\medskip
\noindent
$(2)$ We show that for any $k\in I$, $G\leqv A_k$. Let $x\in G$ and $a\in A_k$. Hence, 
$x\in \cMs(\cup_i A_i)$ and for any $j\in I$, $\cMs(\cup_i A_i) \:\cap \da x \leqv A_j$.
We first prove that $\cMs(\cup_i A_i) \:\cap \da x \subseteq G$. 
Let $y \in \cMs(\cup_i A_i) \:\cap \da x$, and let us check that for any $j\in I$, 
$\cMs(\cup_i A_i) \:\cap \da y \leqv A_j$: if $z\in \cMs(\cup_i A_i) \:\cap \da y $ and $u\in A_j$ then 
$z\in \cMs(\cup_i A_i) \:\cap \da x$ so that $z\vee u\in A_j$ follows, while $z\wedge u \in \cMs(\cup_i A_i) \:\cap \da y$ trivially holds. 
 Now, since $x\wedge a \in \cMs(\cup_i A_i) \:\cap \da x$, we have that $x\wedge a \in G$. On the other hand, 
since $x\in \cMs(\cup_i A_i) \:\cap \da x \leqv A_k$, we also have that $x\vee a \in A_k$.

\medskip
\noindent
$(3)$ We show that if $Z\in \SL(C)$ and, for any $i\in I$, $Z\leqv A_i$ then $Z\leqv G$.  
By point~(1), $\bot = \bigwedge_{i\in I} \bigwedge A_i \in G$.
We then define $Z^\bot \subseteq C$ as follows: $Z^\bot \ud \{ x \vee \bot ~|~ x\in Z\}$. 
It turns out that $Z^\bot \subseteq \cMs(\cup_i A_i)$: in fact, since $C$ is a complete co-Heyting algebra, 
for any $x\in Z$, we have that $x\vee (\bigwedge_{i\in I} \bigwedge A_i) = \bigwedge_{i\in I} (x\vee \bigwedge A_i)$, 
and since $x\in Z$, for any $i\in I$, $\bigwedge A_i \in A_i$, and $Z\leqv A_i$, we have that $x \vee \bigwedge A_i \in A_i$, 
so that  $\bigwedge_{i\in I} (x\vee \bigwedge A_i) \in \cMs(\cup_i A_i)$. 
Also, it turns out that 
$Z^\bot \in \SL(C)$. If $Y\subseteq Z^\bot$ and $Y \neq \varnothing$ then $Y =\{x\vee \bot\}_{x\in X}$ for
some $X\subseteq Z$ with $X\neq \varnothing$. Hence, $\bigvee Y = \bigvee_{x\in X} (x\vee \bot) = (\bigvee X) \vee \bot$,
and since $\bigvee X\in Z$, we therefore have that  $\bigvee Y\in Z^\bot$. On the other hand, $\bigwedge Y = \bigwedge_{x\in X} (x\vee \bot)$,
and, as $C$ is a complete co-Heyting algebra, $\bigwedge_{x\in X} (x\vee \bot) = (\bigwedge X) \vee \bot$, and since 
$\bigwedge X\in Z$, we therefore obtain that $\bigwedge Y \in Z^\bot$. 
We also observe that $Z \leqv Z^\bot$. In fact, if  $x\in Z$ and $y \vee \bot\in Z^\bot$, for some $y\in Z$,  
then, clearly, $x \vee y \vee \bot \in Z^\bot$, while, by distributivity of $C$, $x \wedge (y \vee \bot)= 
(x\wedge y) \vee \bot \in Z^\bot$. 
Next, we show that for any $i\in I$, $Z^\bot \leqv A_i$. Let $x\vee \bot \in Z^\bot$, for some $z\in Z^\bot$, 
and $a\in A_i$. Then, by distributivity of $C$, $(x\vee \bot) \wedge a = (x \wedge a) \vee (\bot \wedge a) = 
(x\wedge a) \vee \bot$, and since, by $Z\leqv A_i$, we know that $x\wedge a \in Z$, we also 
have that $(x\wedge a) \vee \bot\in Z^\bot$. On the other hand, $(x\vee \bot) \vee a = (x \vee a) \vee \bot$, 
and since, by $Z\leqv A_i$, we know that $\bot \leq x\vee a \in A_i$, we obtain 
that $(x\vee a) \vee \bot = x\vee a \in A_i$.

Summing up, we have therefore shown that for any $Z\in \SL(C)$ such that, for any $i\in I$,  $Z\leqv A_i$, 
there exists $Z^\bot \in \SL(C)$ such that $Z^\bot \subseteq \cMs(\cup_i A_i)$ and, for any $i\in I$,
$Z^\bot \leqv A_i$. We now prove that $Z^\bot \subseteq G$. Consider $w\in Z^\bot$, and let us
check that for any $i\in I$, $\cMs(\cup_i A_i)\: \cap {\da w} \leqv A_i$. Hence, consider $y\in \cMs(\cup_i A_i)\:\cap \da w$ and
$a\in A_i$. Then, $y\wedge a\in \cMs(\cup_i A_i)\:\cap \da w$ follows trivially. Moreover, since 
$y\in \cMs(\cup_i A_i)$, there exists a subset $K\subseteq I$, with $K\neq \varnothing$, such that for any 
$k\in K$ there exists $a_k\in A_k$ such that
$y=\bigwedge_{k\in K} a_k$. Thus, since, for any $k\in K$, $z\wedge a_k\in \cMs(\cup_i A_i)\:\cap \da z \leqv A_i$,
we obtain that $\{(z\wedge a_k) \vee a\}_{k\in K} \subseteq A_i$. Since $A_i$ is meet-closed, 
$\bigwedge_{k\in K} \big( (w\wedge a_k) \vee a\big) \in A_i$. 
Since $C$ is a complete co-Heyting algebra, $\bigwedge_{k\in K} \big( (w\wedge a_k) \vee a\big) = 
a \vee \big(\bigwedge_{k\in K} (w\wedge a_k)\big) = a \vee \big( w \wedge (\bigwedge_{k\in K} a_k)\big) = 
a\vee (w\wedge y) = a \vee y$, so that $a\vee y \in A_i$ follows. 

To close the proof of point~(3), we show that $Z^\bot \leqv G$. Let $z \in  Z^\bot$ and $x\in G$. 
On the one hand, since $Z^\bot \subseteq G$, we have that $z\in G$, and, in turn, as $G$ is join-closed, we obtain 
that $z\vee x \in G$. On the other hand, since $x\in \cMs(\cup_i A_i)$, 
there exists a subset $K\subseteq I$, with $K\neq \varnothing$, such that for any 
$k\in K$ there exists $a_k\in A_k$ such that
$x=\bigwedge_{k\in K} a_k$. Thus,  since $Z^\bot \leqv A_k$, for any $k\in K$, we obtain that $z\wedge a_k\in Z^\bot$. 
Hence, since $Z^\bot$ is meet-closed, we have that 
$\bigwedge_{k\in K} (z\wedge a_k) = z \wedge \big(\bigwedge_{k\in K} a_k\big) = z\wedge x
\in Z^\bot$.

\medskip
\noindent
To conclude the proof, we notice that
$\{\top_C\} \in \SL(C)$ is the greatest element in $\tuple{\SL(C),\leqv}$. Thus, since $\tuple{\SL(C),\leqv}$
has nonempty glb's and the greatest element, it 
turns out that it is a complete lattice.
\end{proof}

We have thus shown the following characterization of complete Heyting and co-Heyting algebras.  

\begin{corollary}
Let $C$ be a complete lattice. Then, $\tuple{\SL(C),\leqv}$ is a complete lattice if and only if $C$ is 
a complete Heyting and co-Heyting algebra. 
\end{corollary}

To conclude, we provide an example showing that the property of
being a complete lattice for the poset $\tuple{\SL(C),\leqv}$  cannot be a 
characterization for a complete Heyting (or co-Heyting) algebra $C$.

\begin{example}\rm 
Consider the complete lattice $C$ depicted on the left.  
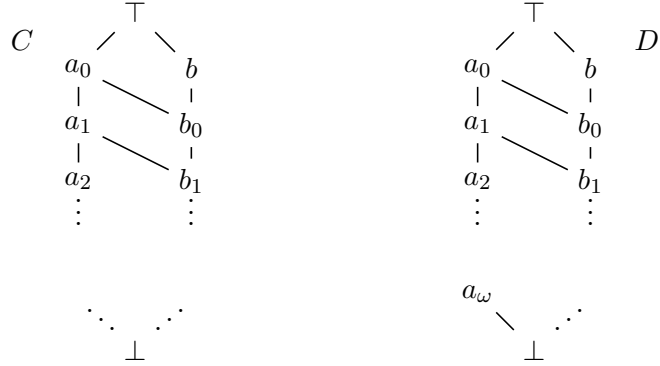
\begin{center}
    \begin{tikzpicture}[scale=0.75]
      \draw (-2,6.5) node[name=top] {{$C$}};

      \draw (0,7) node[name=top] {{$\top$}};
      \draw (-1,6) node[name=a0] {{$a_0$}};
      \draw (1,6) node[name=b] {{$b$}};
      \draw (-1,5) node[name=a1] {{$a_1$}};
      \draw (1,5) node[name=b0] {{$b_0$}};
      \draw (-1,4) node[name=a2] {{$a_2$}};
      \draw (1,4) node[name=b1] {{$b_1$}};
      \draw (0,1) node[name=bot] {{$\bot$}};
      \draw (-1,3.6) node[] {{$\vdots$}};
      \draw (1,3.6) node[] {{$\vdots$}};
      \draw (-0.6,1.7) node[] {{$\ddots$}};
      \draw (0.6,1.7) node[] {{$\iddots$}};

      \draw[semithick] (top) -- (a0);
      \draw[semithick] (top) -- (b);
      \draw[semithick] (b) -- (b0);
      \draw[semithick] (a0) -- (a1);
      \draw[semithick] (a0) -- (b0);
      \draw[semithick] (b0) -- (b1);
      \draw[semithick] (a1) -- (a2);
      \draw[semithick] (a1) -- (b1);

\end{tikzpicture}
\qquad\qquad\qquad\qquad
 \begin{tikzpicture}[scale=0.75]
      \draw (2,6.5) node[name=top] {{$D$}};

      \draw (0,7) node[name=top] {{$\top$}};
      \draw (-1,6) node[name=a0] {{$a_0$}};
      \draw (1,6) node[name=b] {{$b$}};
      \draw (-1,5) node[name=a1] {{$a_1$}};
      \draw (1,5) node[name=b0] {{$b_0$}};
      \draw (-1,4) node[name=a2] {{$a_2$}};
      \draw (1,4) node[name=b1] {{$b_1$}};
      \draw (-1,2) node[name=ao] {{$a_\omega$}};
  
      \draw (0,1) node[name=bot] {{$\bot$}};
      \draw (-1,3.6) node[] {{$\vdots$}};
      \draw (1,3.6) node[] {{$\vdots$}};
      \draw (0.6,1.7) node[] {{$\iddots$}};

      \draw[semithick] (top) -- (a0);
      \draw[semithick] (top) -- (b);
      \draw[semithick] (b) -- (b0);
      \draw[semithick] (a0) -- (a1);
      \draw[semithick] (a0) -- (b0);
      \draw[semithick] (b0) -- (b1);
      \draw[semithick] (a1) -- (a2);
      \draw[semithick] (a1) -- (b1);
      \draw[semithick] (bot) -- (ao);

\end{tikzpicture}

\end{center}

\noindent
$C$ is distributive but not a complete co-Heyting algebra: $b \vee \big(\bigwedge_{i\geq 0} a_i\big) = b < \bigwedge_{i\geq 0} (b\vee a_i) = \top$. 
Let $X_0 \ud \{\top, a_0\}$ and, for any $i\geq 0$, $X_{i+1} \ud X_i \cup \{a_{i+1}\}$, so that $\{X_i\}_{i\geq 0} \subseteq \SL(C)$. 
Then, it turns out that the glb of $\{X_i\}_{i\geq 0}$ in $\tuple{\SL(C),\leqv}$ does not exist. This can be shown by mimicking the proof of Theorem~\ref{main-th}. 
Let $A\ud\{\bot \} \cup \bigcup_{i\geq 0} X_i\in \SL(C)$. Let us observe that $A$ is a
lower bound of $\{X_i\}_{i\geq 0}$. Hence, if we suppose that $Y\in \SL(C)$ is the glb of $\{X_i\}_{i\geq 0}$ 
then $A\leqv Y$ must hold. Hence, if $y\in Y$ then 
$\top \wedge y = y \in A$, so that $Y\subseteq A$, and $\top \vee y \in Y$. Since, $Y\leqv X_0$, we have that $\top \vee y \vee \top =\top \vee y 
\in X_0=\{\top, a_0\}$, so that necessarily $\top \vee y = \top \in Y$. Hence, from $Y\leqv X_i$, for any $i\geq 0$, we obtain that
$\top \wedge a_i = a_i \in Y$. Hence, $Y=A$. The whole complete lattice $C$ is also 
a lower bound of $\{X_i\}_{i\geq 0}$, therefore $C \leqv Y=A$ must hold: however, 
this is a contradiction because from $b\in C$ and $\bot\in A$ we obtain that 
$b\vee \bot = b\in A$.  

\noindent
It is worth noting that if we instead consider the complete lattice $D$ depicted on the right of the above figure, 
which includes a new glb 
$a_\omega$ of the chain $\{a_i\}_{i\geq 0}$, then $D$ becomes a complete Heyting and co-Heyting algebra, and 
in this case the glb of $\{X_i\}_{i\geq 0}$ in $\tuple{\SL(D),\leqv}$ turns out to be $\{\top\} \cup \{a_i\}_{i\geq 0} \cup \{a_\omega\}$.
\qed
\end{example}

\section*{Acknowledgements}
The author has been partially supported by 
the University of Padova under the 2014 PRAT project ``ANCORE''.

\end{document}